\documentclass[conference]{IEEEtran} 
\ifCLASSINFOpdf
\fi
\usepackage{mathtools}
\usepackage{graphicx,color}
\usepackage{cite}
\usepackage{setspace} 
\usepackage{amsmath,amsthm,amssymb,bm}
\usepackage{multirow}
\usepackage{hhline}
\usepackage{epsfig}
\usepackage{epstopdf}
\usepackage{verbatim}
\usepackage{algorithm}
\usepackage{algpseudocode}
\usepackage{cases}
\usepackage{array}
\usepackage{subfigure}
\usepackage{stfloats}
\usepackage{tabstackengine}
\setstackEOL{\cr}

\newtheorem{lemma}{Lemma}

\newtheorem{corollary}{Corollary}

\allowdisplaybreaks

\usepackage{forloop}
\newcounter{ct}

\makeatletter 
\def\@eqnnum{{\normalsize \normalcolor (\theequation)}} 
\makeatother

\makeatletter
\renewcommand*\env@matrix[1][\arraystretch]{%
	\edef\arraystretch{#1}%
	\hskip -\arraycolsep
	\let\@ifnextchar\new@ifnextchar
	\array{*\c@MaxMatrixCols c}}
\makeatother

\begin{document}
	
	\title{Optimally Deployed Multistatic OTFS-ISAC Design With Kalman-Based Tracking of Targets}
	\author{\IEEEauthorblockN{Jyotsna Rani\IEEEauthorrefmark{1}, Kuntal Deka\IEEEauthorrefmark{1}, Ganesh Prasad\IEEEauthorrefmark{2}, and Zilong Liu\IEEEauthorrefmark{3}}\\
		\IEEEauthorblockA{\IEEEauthorrefmark{1}Department of Electronics \& Electrical Engineering, Indian Institute of Technology Guwahati, India\\\IEEEauthorrefmark{2}Department of Electronics and Communication Engineering, National Institute of Technology Silchar, India\\\IEEEauthorrefmark{3}School of Computer Science and Electronics Engineering, University of Essex, Colchester, United Kingdom\\
			emails: \{r.jyotsna, kuntaldeka\}@iitg.ac.in,  gpkeshri@ece.nits.ac.in, zilong.liu@essex.ac.uk
	}}
	
	\maketitle
	\begin{abstract}
		This paper studies orthogonal time-frequency space (OTFS) modulation aided multistatic integrated sensing and communication (ISAC) in vehicular networks, whereby its delay-Doppler robustness is exploited for enhanced communication and high-resolution sensing. We present a triangulation-based deployment framework combined with Kalman filtering (KF) that enables accurate target localization and velocity estimation. In addition, we assess the ISAC performance in the multistatic topology to determine its effectiveness in the dynamic environment. Further, a suboptimal placement strategy for the multistatic receivers is devised to reduce the targets' localization error. Numerical results demonstrate significant reductions in the sensing error and bit error rate (BER) performances.
	\end{abstract}   \vspace{0mm}
	
	
	\vspace{-2mm}
	\section{Introduction}
	Vehicular networks are critical for future intelligent transportation systems, enabling safer, greener, and more efficient driving through enhanced situational awareness \cite{noor22}. Integrated sensing and communication (ISAC) combines shared waveforms, signal processing, and hardware to improve both communication and sensing \cite{FLiu02}. Traditional orthogonal frequency-division multiplexing (OFDM) struggles in high-mobility scenarios due to Doppler-induced interference, whereas orthogonal time-frequency space (OTFS) modulation offers improved communication in dynamic environments \cite{RHad03}. This paper explores OTFS-aided ISAC for vehicular networks, focusing on multistatic ISAC with coordinated, low-complexity nodes to enhance both sensing and communication.

	\subsubsection*{Related Works}
	Monostatic ISAC systems, which feature co-located transmitters and receivers, provide low hardware complexity and efficient spectrum utilization. OFDM-based multipe-input and multiple-output (MIMO) ISAC systems enable reliable estimation of range, angle, and velocity~\cite{su05,che07}. However, a single observation point limits spatial diversity and system robustness, particularly in high-mobility vehicular scenarios. Multistatic ISAC offers a solution by allowing cooperation among multiple nodes, thereby improving both localization accuracy and sensing reliability. Previous studies had shown that joint waveform and resource allocation lead to a lower Cramér–Rao lower bound (CRLB) while maintaining communication performance~\cite{li09}. Cooperative vehicular ISAC frameworks had also been shown to reduce the detection errors and latency. At the network level, strategies such as coordinated multi-point transmission and multi-static radar had been studied to enhance sensing performance as node cooperation increases~\cite{hon11}. Furthermore, mutual-information-based approaches have been used to model multistatic ISAC as a virtual MIMO system~\cite{yin13}. Extending previous developments, authors in \cite{yua21} presented a KF–based tracking approach in multi-static ISAC that integrates bistatic range, range rate, and direction of arrival (DoA) measurements to improve estimation accuracy, noise robustness, and real-time tracking capability. 
    Besides, an OTFS-based vehicle-to-infrastructure multistatic ISAC framework was investigated in \cite{sruti18} that exploits spatial diversity and employs a sparse recovery–driven joint association and localization method to deliver accurate and efficient target detection and parameter estimation in dynamic environments. \textit{Despite recent advances in multistatic ISAC and OTFS-based frameworks, existing studies largely overlook the joint exploitation of OTFS waveforms and geometrical topology for spatially diverse localization and velocity estimation. Moreover, the impact of receiver topology optimization on OTFS-based ISAC and KF-assisted targets' tracking remains insufficiently explored, leaving room for performance enhancement through coordinated geometry-aware design.}
	
	\subsubsection*{Contributions}
	The key contributions in this work is three-fold. (i) We propose a multistatic ISAC system with OTFS waveforms and geometrical triangulation that leverages spatial diversity to significantly improve target localization and velocity estimation. (ii) We show that orthogonal receiver placement maximizes triangulation area, yielding a suboptimal topology that reduces estimation error. (iii) For reliable tracking and active sensing performance, we develop KF-based signal processing algorithms. Furthermore, we investigate a performance analysis of the OTFS-based multistatic ISAC system within the given geometrical topology.  Extensive numerical evaluations demonstrate substantial RMSE reduction and improved bit error rate (BER) performance.
	
	\subsubsection*{Notations}
	Vectors and matrices are denoted by bold lowercase and uppercase letters, respectively. $\mathbb{C}^{M\times N}$ is the set of $M\times N$ complex matrices. $\mathbb{E}\{\cdot\}$ denotes expectation, $\otimes$ the Kronecker product, $\text{vec}(\cdot)$ column-wise vectorization with $\text{vec}^{-1}(\cdot)$ its inverse, $\lVert \cdot \rVert$ the Euclidean norm, $\text{diag}\{\cdot\}$ a diagonal matrix, and $\mathcal{O}(\cdot)$ asymptotic complexity.
    \vspace{-1.5mm}
	
	\section{System Architecture}\label{sec:sys_mod}
	\begin{figure}[!t]
		\centering
		\includegraphics[width=2.5in]{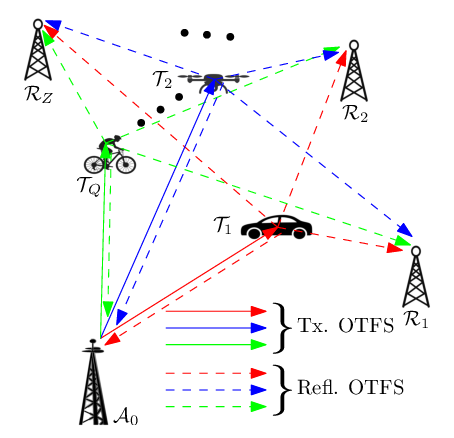}
		\caption{Network topology for multistatic ISAC.}
		\label{fig:ISAC_OTFS_Arch}\vspace{-5mm}
	\end{figure}
	\subsubsection*{System Model}
In Fig.~\ref{fig:ISAC_OTFS_Arch}, the system employs a multistatic ISAC setup with one anchor node (AN) $\mathcal{A}_0$, and several receiver nodes $\{\mathcal{R}_1,\ldots,\mathcal{R}_Z\}$, placed at different locations. The system operates in a multistatic manner by utilizing both monostatic and bistatic configurations among the nodes. In monostatic sensing, $\mathcal{A}_0$ detects targets $\{\mathcal{T}_1,\ldots,\mathcal{T}_Q\}$, by processing echoes from its own transmissions. In bistatic sensing, $\mathcal{A}_0$ transmits while the receivers $\mathcal{R}_z$ ($z \in \{1, \ldots, Z\}$), capture the target reflections. These monostatic and bistatic configurations work together in the multistatic setup, where the cooperation among the nodes enhances the system's sensing capabilities. The sensing data are shared among the nodes through a dedicated wireless channel, enabling cooperative operation for the multistatic sensing. Under active sensing, the anchor sends known probing signals that can be exploited by the receivers for bistatic sensing. Conversely, in passive ISAC operation, the receivers only know the pilot signals for estimation, without knowledge of the transmitted data symbols. To cope with the doubly selective channels, the transmitted signal is modulated using OTFS as described below.
	
	\subsubsection*{OTFS Signalling at the Transmitter}
	At the transmitter, a QAM data vector $\mathbf{d}_I \in \mathbb{A}^{MN}$ is reshaped into a delay-Doppler matrix $\mathbf{D}_I \in \mathbb{A}^{M \times N}$ by splitting it into $N$ segments of length $M$, each forming a column, where $\mathbb{A}$ denotes the QAM alphabet. The parameters $M$ and $N$ denote the numbers of delay and Doppler bins respectively, indexed by $l \in \{0,\dots,M-1\}$ and $k \in \{0,\dots,N-1\}$. The OTFS signal occupies bandwidth of $M\Delta f$ and duration of $NT$ ($T\Delta f=1$), giving delay and Doppler resolutions of $1/M\Delta f$ and $1/NT$, with $(l,k)$ corresponding to the delay of $l/M\Delta f$ and the Doppler shift of $k/NT$. For ISAC, a pilot matrix $\mathbf{D}_P \in \mathbb{C}^{M \times N}$ is superimposed to the data matrix. Such a pilot matrix has a single nonzero entry at $(l_p,k_p)$, given by  
	$\mathbf{D}_P(l_p,k_p)=\sqrt{MN\sigma_P^2}$, 
	ensuring average pilot power of $\sigma_P^2$. Therefore, the combined transmit matrix is $\mathbf{D}=\mathbf{D}_I+\mathbf{D}_P$. Then, the delay-Doppler symbols are mapped to the time-frequency domain via the inverse symplectic fast Fourier transform (ISFFT) $\mathbf{X}=\mathbf{F}_M \mathbf{D} \mathbf{F}_N^H$, where $\mathbf{F}_M$ and $\mathbf{F}_N$ are fast Fourier transform (FFT) square matrices of size $M$ and $N$, respectively. After the inverse FFT and pulse shaping with $\mathbf{P}_\text{tx}=\text{diag}\{p_\text{tx}(0),\dots,p_\text{tx}(M-1)\}$, 
	the transmit matrix is  
	\begin{align}\label{eq:tx_mat}
		\mathbf{S}=\mathbf{P}_\text{tx}\,\mathbf{F}_M^H\mathbf{X}.
	\end{align}
    Using \eqref{eq:tx_mat}, the time-domain transmit vector $\mathbf{s}$ in \eqref{eq:tx_vec} is derived by substituting $\mathbf{X}=\mathbf{F}_M\mathbf{D}\mathbf{F}_N^H$ into the expression for $\mathbf{S}$ and subsequently applying the vectorization property $\mathrm{vec}(ABC)=(C^T\!\otimes\!A)\mathrm{vec}(B)$.
	\begin{align}\label{eq:tx_vec}
		\mathbf{s} = (\mathbf{F}_N^H \otimes \mathbf{P}_\text{tx}) \mathbf{d}, \quad \mathbf{d}=\mathrm{vec}(\mathbf{D}).
	\end{align}
	
	\subsubsection*{Received Signal Processing}
	In a doubly-selective channel, the received signal at the $j$-th receiver is the superposition of $P$ propagation paths, each introducing a complex gain, cyclic delay, and Doppler shift. The channel matrix is given as 
	\begin{align}
		\mathbf{H}_j = \sum_{q=1}^P h_{q,j}\,\mathbf{\Pi}^{l_{q,j}}\mathbf{\Delta}^{k_{q,j}},
	\end{align}
	where $h_{q,j}$ is the path gain, $l_{q,j}$ and $k_{q,j}$ denote the delay and Doppler indices, $\mathbf{\Pi}$ is the cyclic shift matrix (cf. \cite[eq. (9)]{rav01}), and $\mathbf{\Delta}=\mathrm{diag}\{c^0,\dots,c^{MN-1}\}$ with $c=\mathrm{e}^{i2\pi/MN}$. Thus, the received signal is   
	\begin{align}
		\mathbf{r}_j = \mathbf{H}_j\mathbf{s} + \mathbf{n}_j, \quad \mathbf{n}_j\sim\mathcal{CN}(\mathbf{0},\sigma_j^2\mathbf{I}_{MN}).
	\end{align}
	Transforming $\mathbf{r}_j$ into the delay-Doppler domain, we obtain  
	\begin{align}
		\mathbf{y}_j = (\mathbf{F}_N \otimes \mathbf{P}_{rx})\mathbf{r}_j 
		= \overline{\mathbf{H}}_j\mathbf{d} + \overline{\mathbf{n}}_j,
	\end{align}
	with  
	$\overline{\mathbf{H}}_j = \sum_{q=1}^P h_{q,j}\,\mathcal{T}(l_{q,j},k_{q,j})$ and 
	$\mathcal{T}(l,k) = (\mathbf{F}_N \otimes \mathbf{P}_{rx}) \mathbf{\Pi}^l\mathbf{\Delta}^k (\mathbf{F}_N^H \otimes \mathbf{P}_{tx})$. Here, $\mathcal{T}(l,k)$ represents the joint delay–Doppler transformation operator that imposes a cyclic time shift by $l$ and a frequency modulation by $k$ on the transmitted signal. For rectangular pulses, $\mathbf{P}_{tx}=\mathbf{P}_{rx}=\mathbf{I}_M$. The equivalent noise is $\overline{\mathbf{n}}_j\sim\mathcal{CN}(\mathbf{0},\sigma_j^2\mathbf{I}_{MN})$.  

\vspace{-1mm}
\subsubsection*{Active Sensing Procedure}
In active sensing, the transmit vector $\mathbf{d}$ is known to all receivers. The target range and velocity at the $j$th receiver are determined by estimating the complex gains $\mathbf{h}_j$, delays $\bm{\tau}_j$, and Doppler shifts $\bm{\nu}_j$ using maximum likelihood estimation, which simplifies to a least-squares problem under Gaussian noise:
	\begin{align}
		\!\!\!\!  (\widehat{\mathbf{h}}_j,\widehat{\bm{\tau}}_j,\widehat{\bm{\nu}}_j) 
		= \arg\min_{(\mathbf{h}_j,\bm{\tau}_j,\bm{\nu}_j)} 
		\Big\lVert \mathbf{y}_{j}\hspace{-1mm}-\hspace{-1mm}\sum_{q=1}^P h_{q,j}\mathcal{T}(l_{q,j},k_{q,j})\mathbf{d} \Big\rVert^2\hspace{-1.5mm}.\!\!\!\!  
	\end{align}
To address the non-convex and high-dimensional search space, a \textit{sequential approach} is used. Initially, the dominant path is identified by maximizing the correlation between the received signal and a delayed-Doppler shifted version of $\mathbf{d}$. Subsequent paths are estimated iteratively by applying \textit{interference mitigation}: previously estimated path contributions are subtracted before estimating the next path's parameters. The estimation of the $i$th path is given by:
	\begin{align}\label{eq:tau_n_nu}\nonumber
		(\widehat{\tau}_{i,j},\widehat{\nu}_{i,j}) &= \arg\max_{(\tau_{i,j},\nu_{i,j})}\left| (\mathcal{T}(\tau_{i,j},\nu_{i,j})\mathbf{d})^H\left(\mathbf{y}_j \right.\right.\\
        &\;\;\left.\left.-\sum_{q=1}^{i-1}\widehat{h}_{q,j}\mathcal{T}(\widehat{\tau}_{q,j},\widehat{\nu}_{q,j})\mathbf{d}\right) \right|^2.
	\end{align}
The corresponding gain for the $i$th path is then estimated as:
	\begin{align}\label{eq:gain_est}
		\widehat{h}_{i,j} = \frac{(\mathcal{T}(\widehat{\tau}_{i,j},\widehat{\nu}_{i,j})\mathbf{d})^H\mathbf{y}_j}{\mathbf{d}^H\mathcal{T}(\widehat{\tau}_{i,j},\widehat{\nu}_{i,j})^H\mathcal{T}(\widehat{\tau}_{i,j},\widehat{\nu}_{i,j})\mathbf{d}}.
	\end{align}  
This \textit{iterative path estimation with interference mitigation} ensures accurate estimation of each reflected path's parameters by suppressing interference from previously detected signals.  
	\subsubsection*{Computational Complexity}
	The cost is dominated by the two-dimensional delay--Doppler search, which requires $\mathcal{O}(G_\tau G_\nu MN)$ operations per path at one receiver, where $G_\tau$ and $G_\nu$ denote the grid sizes in delay and Doppler, respectively. For the $P$ paths and $Z+1$ receivers, the overall complexity is $\mathcal{O}((Z+1)P G_\tau G_\nu MN)$, with the correlation search as the primary bottleneck.  

\vspace{-0.5mm}
    \subsubsection*{Computation of Range and Velocity}
	Finally, the self-reflected signal is received by receiver $\mathcal{A}_0$, whereby the range $\rho_{i,0}$ and the radial velocity $v_{i,0}$ of the $i$th target can be computed as in \eqref{eq:est_self} below. For receiver $\mathcal{R}_j$, the range $\rho_{i,j}$ and radial velocity $v_{i,j}$ of the $i$th target are given in \eqref{eq:est_ref}.
	\begin{subequations}\label{eq:est_tar_para}
		\begin{align}\label{eq:est_self}
			&\widehat{\rho}_{i,0} = \widehat{\tau}_{i,0}c/2;\;\; \widehat{v}_{i,0} = \widehat{\nu}_{i,0}c/2f_c,\\\label{eq:est_ref}
			& \widehat{\rho}_{i,j} = \widehat{\tau}_{i,j}c-\widehat{\rho}_{i,0}; \;\; \widehat{v}_{i,j} = (\widehat{\nu}_{i,j} - \widehat{\nu}_{i,0})c/f_c.
		\end{align}
	\end{subequations}

    \vspace{-0.5mm}
	\subsubsection*{ISAC Procedure}
	In this scenario, only the pilot signal is known at the receivers. The channel, defined by its gains, delays, and Doppler shifts, is first estimated. After the initial channel estimation, the information data symbols are detected, and the channel is iteratively refined with both the detected data and the pilot. For a coarse estimate, the pilot delay-Doppler matrix $\mathbf{D}_P$ is vectorized as $\mathbf{d}_P = \text{vec}(\mathbf{D}_P)$. Each receiver treats the transmitted data $\mathbf{d}_I$ as interference and estimates the channel via
	\begin{align}\label{eq:est_pilot_min}(\widehat{\mathbf{h}}_j^{(0)},\widehat{\bm{\tau}}_j^{(0)},\widehat{\bm{\nu}}_j^{(0)})\hspace{-0.5mm} = \hspace{-0.5mm}\arg\hspace{-0.5mm}\min_{(\mathbf{h}_j,\bm{\tau}_j,\bm{\nu}_j)}\hspace{-0.5mm} \left\lVert \mathbf{y}_{j}\hspace{-0.5mm}-\hspace{-0.5mm}\sum_{q=1}^P\hspace{-0.5mm} h_{q,j}\hspace{-0.5mm}\mathcal{T}(\tau_{q,j},\hspace{-0.5mm}\nu_{q,j})\mathbf{d}_P \right\rVert^2.
	\end{align}
	Sequential estimation as described in \eqref{eq:tau_n_nu} and \eqref{eq:gain_est} can be applied, though accuracy is lower than in active sensing where $\mathbf{d}_I$ is known. To refine the channel, the information data is recovered via a regularized least-squares problem
	\begin{align}\label{eq:reg_data_est}
		\min_{\mathbf{d}} \lVert \mathbf{y}_j - \overline{\mathbf{H}}_j\mathbf{d} \rVert^2 + \mu \lVert \mathbf{d} \rVert^2,
	\end{align}  
	with a linear MMSE solution below 
	\begin{align}\label{eq:closed_sol}
		\widehat{\mathbf{d}}_{\text{MMSE}} = (\overline{\mathbf{H}}_j^H\overline{\mathbf{H}}_j + \mu\mathbf{I})^{-1}\overline{\mathbf{H}}_j^H\mathbf{y}_j.
	\end{align}  
	For large channel matrices, gradient descent can be used as
	\begin{subequations}\label{eq:gradient}
		\begin{align}
			\nabla J(\mathbf{d})^{(t)} &= 2\overline{\mathbf{H}}_j^H(\overline{\mathbf{H}}_j\mathbf{d}^{(t)} - \mathbf{y}_j) + 2\mu\mathbf{d}^{(t)},\\
			\mathbf{d}^{(t+1)} &= \mathbf{d}^{(t)} - \eta\nabla J(\mathbf{d})^{(t)},
		\end{align}
	\end{subequations}  
	until $\lVert \mathbf{d}^{(t+1)} - \mathbf{d}^{(t)} \rVert < \epsilon$. The estimated vector $\mathbf{d}'$ is then transformed back to the delay-Doppler domain: $\mathbf{D}' = \text{vec}^{-1}(\mathbf{d}')$, and the information component is extracted by $\mathbf{D}_I' = \mathbf{D}' - \mathbf{D}_P$. Symbol-wise demodulation is then performed as
	\begin{align}\label{eq:symb_demod}
		\widehat{\mathbf{D}}_I = \arg\min_{\mathbf{D}_I \in \mathbb{A}^{M \times N}} \lVert \mathbf{D}_I - \mathbf{D}_I' \rVert^2.
	\end{align}  
	Using both $\mathbf{D}_P$ and $\widehat{\mathbf{D}}_I$, the channel is refined via the active sensing and data are re-estimated. This outer loop alternates between data detection and channel refinement until convergence. The final channel estimates at all $Z+1$ receivers yield target ranges and radial velocities via \eqref{eq:est_tar_para}, enabling localization and velocity estimation.

    \vspace{-2mm}
	\section{Targets Localization Via Cooperative Sensing}\label{sec:loc_n_sen}
	\begin{figure}[!t]
		\centering  \includegraphics[width=2.0in]{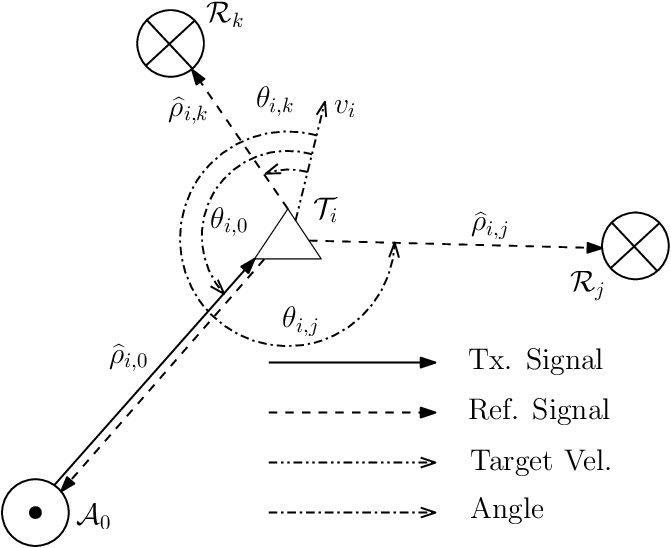}
		\caption{Estimation of location and velocity of a target via cooperative sensing.}    
		\label{fig:dist_vel_est}\vspace{-2mm}
	\end{figure}
	
	This section presents a method to estimate a target’s position and speed, refining accuracy using estimated ranges and radial velocities relative to AN $\mathcal{A}_0$ and multiple receivers. In the analysis, we assume that all nodes are approximately coplanar, thus ignoring elevation differences among them, which is a practical simplification in most surface-level sensing scenarios.
	
	\subsubsection*{Targets' Location}
	As shown in Fig.~\ref{fig:dist_vel_est}, the position of target $\mathcal{T}_i$ is estimated via geometric triangulation using the transmitter $\mathcal{A}_0$ and two distinct receivers $\mathcal{R}_j$ and $\mathcal{R}_k$ ($j, k \in \{1, \dots, Z\}, j \ne k$). The ranges $\widehat{\rho}_{i,0}$, $\widehat{\rho}_{i,j}$, and $\widehat{\rho}_{i,k}$ are computed according to \eqref{eq:est_tar_para}. With the nodes' known positions $(x_0, y_0)$, $(x_j, y_j)$, and $(x_k, y_k)$, the target's coordinates $(\alpha_i, \beta_i)$ follow from the Euclidean distance equations in \eqref{eq:euclidean_dist}. Subtracting the equations for $q=j$ and $q=k$ from that for $q=0$ gives the simplified linear system in \eqref{eq:simplified}. Further, these are represented in matrix form in \eqref{eq:matrix_form}, as shown below
	\begin{subequations}
		\begin{align}\label{eq:euclidean_dist}
			&(x_q - \alpha_i)^2 + (y_q - \beta_i)^2 = \rho_{i,q}^2;\;\;\;\;\;\; q\in\{0,j,k\}\\\nonumber
			&\begin{bmatrix}
				x_q - x_0 & y_q - y_0
			\end{bmatrix}
			\begin{bmatrix}
				\alpha_i\\ 
				\beta_i
			\end{bmatrix}
			=\frac{1}{2}(\rho_{i,0}^2 - \rho_{i,q}^2 \\\label{eq:simplified}
			&\hspace{0.7in} - (x_0^2 - x_q^2) - (y_0^2 - y_q^2)); \;\;\;\; q \in \{j,k\}\\
			& \mathbf{A}
			\begin{bmatrix}
				\alpha_i\\ 
				\beta_i
			\end{bmatrix}
			= \mathbf{B};\;\;\;\; \mathbf{A} = 
			\begin{bmatrix}
				x_j - x_0 & y_j - y_0 \\\nonumber
				x_k - x_0 & y_j - y_0
			\end{bmatrix}\\\label{eq:matrix_form}
			& \mathbf{B} = \frac{1}{2}
			\begin{bmatrix}
				\rho_{i,0}^2 - \rho_{i,j}^2 - (x_0^2 - x_j^2) - (y_0^2 - y_j^2)) \\
				\rho_{i,0}^2 - \rho_{i,k}^2 - (x_0^2 - x_k^2) - (y_0^2 - y_k^2))
			\end{bmatrix}\\\label{eq:solution}
			& \begin{bmatrix}
				\alpha_i\\ 
				\beta_i
			\end{bmatrix}
			= \mathbf{A}^{-1}\mathbf{B}.
		\end{align}
	\end{subequations}
	The solution in \eqref{eq:solution} uniquely determines the $i$th target’s location if $\mathbf{A}$ is full rank, which occurs when the transmitter and two receivers are non-collinear.
	
	\subsubsection*{Targets' Velocity}
	Using geometric triangulation in Fig.~\ref{fig:dist_vel_est}, we estimate the target's velocity. The unit vector from the $i$th target to node $q \in \{j, k\}$ is $\mathbf{u}_{q,i} = \frac{[x_q - \alpha_i,\, y_q - \beta_i]^T}{\rho_{i,q}}$, where $(x_q, y_q)$ are the node coordinates and $\rho_{i,q}$ is the distance to the target. If the target's velocity is $\mathbf{v}_i = [v_{i,x}, v_{i,y}]^T$, its projection onto $\mathbf{u}_{q,i}$ gives the radial velocity measured by node $q$: $v_{i,q} = \mathbf{v}_i^T \mathbf{u}_{q,i}, \quad q \in \{j, k\}$. By arranging these equations in matrix form, we obtain a linear system as
	\begin{align}\nonumber\label{eq:vel_cal}
		&\mathbf{C}
		\begin{bmatrix}
			v_{i,x} \\
			v_{i,y}
		\end{bmatrix}
		= \mathbf{D}
		\quad \Rightarrow \quad
		\begin{bmatrix}
			v_{i,x} \\
			v_{i,y}
		\end{bmatrix}
		= \mathbf{C}^{-1}\mathbf{D}; \\
		&\mathbf{C} = 
		\begin{bmatrix}
			(x_j - \alpha_i)/\rho_{i,j} & (y_j - \beta_i)/\rho_{i,j} \\
			(x_k - \alpha_i)/\rho_{i,k} & (y_k - \beta_i)/\rho_{i,k}
		\end{bmatrix},\;\;
		\mathbf{D} = 
		\begin{bmatrix}
			v_{i,j} \\
			v_{i,k}
		\end{bmatrix}.
	\end{align}
	Assuming $\mathbf{C}$ is invertible, the system gives the velocity vector $\mathbf{v}_i$. With $Z$ receivers, $\mathcal{Z} \triangleq \binom{Z}{2}$ unique pairs are combined with $\mathcal{A}_0$ to form triangles, each providing an independent estimate of the target's location and velocity via \eqref{eq:solution} and \eqref{eq:vel_cal}.
	
	\subsubsection*{Nearest-Neighbor-Based Selection Approach}
	Averaging all $\mathcal{Z}$ triangle-based estimates reduces random noise but ignores the reliability of each configuration. Some triangles can produce inaccurate estimates, for example, when nodes are nearly collinear or too close. To address this, we describe a \textit{nearest-neighbor selection method} that favors consistent estimates. Let the set of estimates for the $i$th target be $\mathcal{C} = \{C_1, \dots, C_{\mathcal{Z}}\}$, with each $C_q \in \mathbb{R}^2$ representing $[\alpha_{i,q}, \beta_{i,q}]^T$ or $[v_{i,x,q}, v_{i,y,q}]^T$. The most coherent subset is identified by computing pairwise Euclidean distances within $\mathcal{C}$. For a predefined threshold $\xi > 0$, the neighbor set $\mathcal{S}_n$ of an estimate $C_n$ includes all other estimates that lie within distance $\xi$ of it:
	\begin{align}\label{eq:neigh_set}
		\mathcal{S}_n = \left\{ C_m \in \mathcal{C} \setminus \{C_n\} \;\middle|\; \lVert C_n - C_m \rVert \le \xi \right\}.
	\end{align}
	Among all such sets, we identify the one with the largest number of neighbors:
	\begin{align}
		n^* = \arg \max_n |\mathcal{S}_n|.
	\end{align}
	where $|\cdot|$ denotes the number of elements in a set. The final estimate is then computed by averaging $C_{n^*}$ with all points in its neighbor set:
	\begin{align}\label{eq:nearest_sel}
		C_{\text{avg}} = \frac{1}{|\mathcal{S}_{n^*}| + 1} \left( C_{n^*} + \sum_{C_m \in \mathcal{S}_{n^*}} C_m \right).
	\end{align}
	By averaging only estimates that are spatially consistent, this method minimizes the impact of large errors. We then examine how the positions of AN $\mathcal{A}_0$ and receivers influence target location estimation using geometric triangulation.

	\subsubsection*{Suboptimal Locations of the Receivers}
	Using Fig.~\ref{fig:dist_vel_est}, we find the effect of the location of the AN and receivers on the estimation of the targets' location using the following lemma.
	\begin{lemma}\label{lemma1}
		For the given variances, $\{\sigma_{i,q}^2 \mid q \in \{0, j, k\}\}$ in the estimation of the square of the ranges of the $i$th target from the nodes, $\mathcal{A}_0$, $\mathcal{R}_j$, and $\mathcal{R}_k$, one can minimize $\operatorname{tr}(\operatorname{cov}(\widehat{\alpha}_{i},\widehat{\beta}_{i}))$ or $\kappa_{\max}(\operatorname{cov}(\widehat{\alpha}_{i},\widehat{\beta}_{i}))$ of the estimated location $(\widehat{\alpha}_{i},\widehat{\beta}_{i})$ using geometrical triangulation (cf. Fig.~\ref{fig:dist_vel_est}) by optimizing the magnitude of the area of the triangle formed using the three nodes. 
	\end{lemma}
	\begin{proof}
		Without loss of generality, for simplicity, we assume that the AN $\mathcal{A}_0$ is placed at the origin, i.e., its coordinate is $(0,0)$ and the coordinates of $\mathcal{R}_j$ and $\mathcal{R}_k$ are $(x_j,y_j)$ and $(x_k,y_k)$.  Using \eqref{eq:euclidean_dist}, it can be shown that for given $\{\sigma_{i,q}^2 \mid q \in \{0, j, k\}\}$, the variances $\sigma_{\alpha_i}^2$ and $\sigma_{\beta_i}^2$ in the estimation of the x and y coordinates of the location of the $i$th target using the geometrical triangulation in Fig.~\ref{fig:dist_vel_est} is given by
		\begin{align}\nonumber
			\sigma_{\gamma_i}^2 = \frac{\sigma_{i,j}^2(g_j+g_k)^2 + \sigma_{i,0}^2g_k^2 + \sigma_{i,k}^2g_j^2}{\mathcal{D}},
		\end{align}
		where $\mathcal{D} = (x_j y_k - x_k y_j)^2$, the variable $g$ is chosen as $g = y$ when computing $\sigma_{\alpha_i}^2$ (x-coordinate variance), and $g = x$ when computing $\sigma_{\beta_i}^2$ (y-coordinate variance). We assume that the errors in the estimation of $\widehat{\alpha}_i$ and $\widehat{\beta}_i$ are independent of each other. Therefore, the covariance matrix $\operatorname{cov}(\widehat{\alpha}_{i},\widehat{\beta}_{i})$ is diagonal matrix. Thus, $\operatorname{tr}(\operatorname{cov}(\widehat{\alpha}_{i},\widehat{\beta}_{i}))$ and $\kappa_{\max}(\operatorname{cov}(\widehat{\alpha}_{i},\widehat{\beta}_{i}))$ are:
		\begin{subequations}
			\begin{align}\nonumber
				&\operatorname{tr}(\operatorname{cov}(\widehat{\alpha}_{i},\widehat{\beta}_{i})) = \sigma_{\alpha_i}^2 + \sigma_{\beta_i}^2 = \sigma_{i,j}^2[(x_j+x_k)^2 + (y_j+y_k)^2] \\\label{eq:trace_cov}
				& + \sigma_{i,0}^2[x_k^2 + y_k^2] + \sigma_{i,k}^2[x_j^2 + y_j^2]/(x_j y_k - x_k y_j)^2, \\\nonumber
				& \kappa_{\max}(\operatorname{cov}(\widehat{\alpha}_{i},\widehat{\beta}_{i})) = \frac{1}{(x_j y_k - x_k y_j)^2}\max\{\sigma_{i,j}^2(x_j+x_k)^2 \\\label{eq:eigen_cov}
				& + \sigma_{i,0}^2x_k^2 + \sigma_{i,k}^2x_j^2,\;\; \sigma_{i,j}^2(y_j\hspace{-0.5mm}+\hspace{-0.5mm}y_k)^2 \hspace{-0.5mm} + \hspace{-0.5mm} \sigma_{i,0}^2y_k^2 \hspace{-0.5mm} + \hspace{-0.5mm} \sigma_{i,k}^2y_j^2\}.
			\end{align}
		\end{subequations}
		Here, the denominator of \eqref{eq:trace_cov} and \eqref{eq:eigen_cov} is proportional to the triangle's area formed by the three nodes, so maximizing this area within a bounded region is desired. The numerators depend on the x and y coordinates of the receivers, with smaller values decreasing the area. Therefore, optimizing the triangle's area minimizes the trace and maximum eigenvalue.
	\end{proof}
	\noindent One of the sub-optimal solution for the \textbf{Lemma}~\ref{lemma1} can be obtained using the following corollary.
	\begin{corollary}\label{coro1}
		For the given variances, $\{\sigma_{i,q}^2\}$; $q\in \{0,j,k\}$, if the receivers are placed at orthogonal axis in \textbf{Lemma}~\ref{lemma1}, then, $\operatorname{tr}(\operatorname{cov}(\widehat{\alpha}_{i},\widehat{\beta}_{i}))$ or $\kappa_{\max}(\operatorname{cov}(\widehat{\alpha}_{i},\widehat{\beta}_{i}))$ can be minimized by maximizing the area of the triangle formed using the three nodes in a given bounded region. 
	\end{corollary}
	\begin{proof}
		If the receivers $\mathcal{R}_j$ and $\mathcal{R}_k$ are placed at the x and y axes (orthogonal axes) of the Cartesian coordinates in \textbf{Lemma}~\ref{lemma1}, then their coordinates are given by: $(x_j,0)$ and $(0,y_k)$, respectively. After, substituting these coordinates in \eqref{eq:trace_cov} and \eqref{eq:eigen_cov}, we obtain
		\begin{subequations}\label{eq:trace_eigen_three}
			\begin{align}\label{eq:trace_sp}
				&\!\!\!\!\operatorname{tr}(\operatorname{cov}(\widehat{\alpha}_{i},\widehat{\beta}_{i})) = \frac{\sigma_{i,j}^2 + \sigma_{i,k}^2}{y_k^2} + \frac{\sigma_{i,j}^2 + \sigma_{i,0}^2}{x_j^2},\!\!\!\! \\\label{eq:eigen_sp}
				&\!\!\!\!\kappa_{\max}(\operatorname{cov}(\widehat{\alpha}_{i},\widehat{\beta}_{i})) = \max\left\{\frac{\sigma_{i,j}^2 + \sigma_{i,0}^2}{x_j^2}, \frac{\sigma_{i,j}^2 + \sigma_{i,k}^2}{y_k^2}\right\}.\!\!\!\!
			\end{align}
		\end{subequations}
		From \eqref{eq:trace_sp} and \eqref{eq:eigen_sp}, $\operatorname{tr}(\operatorname{cov}(\widehat{\alpha}_{i},\widehat{\beta}_{i}))$ and $\kappa_{\max}(\operatorname{cov}(\widehat{\alpha}_{i},\widehat{\beta}_{i}))$ can be minimized by maximizing the magnitude of $x_j$ and $y_k$ which is equivalent to maximizing the magnitude of the area of right angled triangle formed using the three nodes.
	\end{proof}
    \noindent Further, in the extended journal version of this work \cite{rani20}, \textbf{Corollary}~\ref{coro1} states that placing multiple receivers along orthogonal axes yields a suboptimal configuration in which the area formed by each triangulation is maximum. This is equivalent to a setup with three nodes: one AN and two receivers, each equipped with multiple independent antennas.
    
    \vspace{-2mm}
	\section{Multistatic Sensing and ISAC}
	This section presents trajectory estimation of moving targets using correlated random walk (CRW) based KF, reflecting the tendency of humans, vehicles, and UAVs to maintain their moving directions for a period instead of moving randomly. Next, we describe the multistatic ISAC.
	
	\subsubsection*{CRW-Based Prediction Model} The CRW leverages the Ornstein-Uhlenbeck process to model target motion \cite{joh02}. The velocity $\mathbf{v}_{i,t} = [v_{i,\alpha,t}, v_{i,\beta,t}]^T$ of the $i$th target at time $t$ is:
	\begin{align}
		\mathrm{d}\mathbf{v}_{i,t} = \delta_i(\mathbf{v}_{i,t} - \bm{\omega}_{i})\mathrm{d}t + \psi_i \mathrm{d}\mathbf{W}_{i,t},
	\end{align}
	where $\delta_i$ is the velocity autocorrelation, $\bm{\omega}_i$ the mean velocity, $\psi_i$ the diffusion coefficient, and $\mathbf{W}_{i,t}$ a Wiener process. Its solution in continuous and discrete time is:
	\begin{align}
		&\mathbf{v}_{i,t} = \bm{\omega}_i + (\mathbf{v}_{i,0} - \bm{\omega}_i)e^{-\delta_i t} + \psi_i \int_{0}^t e^{-\delta_i (t-s)}\mathrm{d}\mathbf{W}_{i,s},\\
		&\mathbf{v}_{i,t+1} = e^{-\delta_i \Delta t}\mathbf{v}_{i,t} + \mathbf{n}_{i,t}^{\mathbf{v}},
	\end{align}
	with $\mathbf{n}_{i,t}^{\mathbf{v}}$ representing Gaussian velocity noise. The target position $\mathbf{l}_{i,t} = [\alpha_{i,t}, \beta_{i,t}]^T$ is updated by
	\begin{align}
		&\mathbf{l}_{i,t+1} = \mathbf{l}_{i,t} + \frac{1 - e^{-\delta_i \Delta t}}{\delta_i}\mathbf{v}_{i,t} + \mathbf{n}_{i,t}^{\mathbf{l}},
	\end{align}
	where $\mathbf{n}_{i,t}^{\mathbf{l}}$ is Gaussian location noise.
	
	Using the CRW model in a KF, the state of the $i$th target is $\mathbf{S}_{i,t} = [\alpha_{i,t}, \beta_{i,t}, v_{i,\alpha,t}, v_{i,\beta,t}]^T$, with discrete-time transition:
	\begin{subequations}
		\begin{align}\label{eq:state_tran}
			&\mathbf{s}_{i,t+1} = \mathbf{T}_i \mathbf{s}_{i,t} + \mathbf{n}_{i,t}^{\mathbf{S}},\\
			&\mathbf{T}_i =
			\begin{bmatrix}
				1 & 0 & \frac{1 - e^{-\delta_i \Delta t}}{\delta_i} & 0 \\
				0 & 1 & 0 & \frac{1 - e^{-\delta_i \Delta t}}{\delta_i} \\
				0 & 0 & e^{-\delta_i \Delta t} & 0 \\
				0 & 0 & 0 & e^{-\delta_i \Delta t} 
			\end{bmatrix},
		\end{align}
	\end{subequations}
	where $\mathbf{n}_{i,t}^{\mathbf{S}}$ is Gaussian noise with covariance $\mathbf{Q}_{i,t} = \mathbb{E}[\mathbf{n}_{i,t}^{\mathbf{S}}{\mathbf{n}_{i,t}^{\mathbf{S}}}^T]$.
	\vspace{-0.5mm}
	\subsubsection*{Measurement Model and Trajectory Estimation Using KF}
	As discussed in Section~\ref{sec:loc_n_sen}, the measurement $\mathbf{z}_{i,t}$ for the $i$th target’s position and velocity is obtained using the selection approach in \eqref{eq:nearest_sel}. The measurement model is given by:  
	\begin{align}\label{eq:meas_mod}
		\mathbf{z}_{i,t} = \mathbf{V}_i \mathbf{s}_{i,t} + \mathbf{n}_{i,t}^\mathbf{z},
	\end{align}
	where $\mathbf{V}_i$ is the observation matrix and $\mathbf{n}_{i,t}^\mathbf{z}$ denotes zero-mean Gaussian noise with covariance $\mathbf{R}_i = \mathbb{E}[\mathbf{n}_{i,t}^\mathbf{z}{\mathbf{n}_{i,t}^\mathbf{z}}^T]$. Since both $x$- and $y$-components of location and velocity are considered, $\mathbf{V}_i = \mathbf{I}_4$. Using \eqref{eq:state_tran} and \eqref{eq:meas_mod}, the KF performs prediction and correction to enhance state estimation accuracy. In \textbf{Algorithm}~\ref{algo11}, the initialization includes $\mathbf{T}_i$, $\mathbf{V}_i$, process covariance $\mathbf{Q}_{i,0}$, measurement $\mathbf{z}_{i,1}$ with noise covariance $\mathbf{R}_{i,1}$, and the initial state estimate $\mathbf{s}^a_{i,0} = \mathbb{E}[\mathbf{s}_{i,0}]$ with error covariance $\mathbf{P}_{i,0}$. At each time $t$, the algorithm as a function $f_{KF}$ (as described in Algorithm \ref{algo11}) uses $\mathbf{Q}_{i,t-1}$, $\mathbf{R}_{i,t}$, $\mathbf{s}^a_{i,t-1}$, $\mathbf{P}_{i,t-1}$, and $\mathbf{z}_{i,t}$ to compute $\mathbf{s}^f_{i,t}$, $\mathbf{P}^f_{i,t}$, $\mathbf{K}_{i,t}$, $\mathbf{s}^a_{i,t}$, and $\mathbf{P}_{i,t}$, yielding the updated $\mathbf{s}^a_{i,t}$ and $\mathbf{P}_{i,t}$.
	
	\begin{algorithm}[!t]
		\caption{KF function $f_{KF}$ to estimate the state $\mathbf{s}_{i,t}^a$.}\label{algo11}
		\begin{algorithmic}[1]
			\State System parameters $\mathbf{T}_i$ and $\mathbf{V}_i$
			\State \textbf{Input:} $\mathbf{Q}_{i,t-1}$, $\mathbf{R}_{i,t}$, $\mathbf{s}^a_{i,t-1}$, $\mathbf{P}_{i,t-1}$, and the measurement $\mathbf{z}_{i,t}$
			
			\State Determine the forecasted state: $\mathbf{s}^f_{i,t} = T_i \mathbf{s}^a_{i,t-1}$\label{step:3}
			\State Update the forecast error covariance: $\mathbf{P}^f_{i,t} = \mathbf{T}_i \mathbf{P}_{i,t-1} \mathbf{T}_i^T + \mathbf{Q}_{i,t-1}$\label{step:4}
			
			\State Determine Kalman gain: $\mathbf{K}_{i,t} = \mathbf{P}^f_{i,t} \mathbf{V}_i^T (\mathbf{V}_i \mathbf{P}^f_{i,t} \mathbf{V}_i^T + \mathbf{R}_{i,t})^{-1}$\label{step:5}
			\State Update state estimate: $\mathbf{s}^a_{i,t} = \mathbf{s}^f_{i,t} + \mathbf{K}_{i,t} (\mathbf{z}_{i,t} - \mathbf{V}_i \mathbf{s}^f_{i,t})$\label{step:6}
			\State Update posterior error covariance: $\mathbf{P}_{i,t} = (\mathbf{I} - \mathbf{K}_{i,t} \mathbf{V}_i) \mathbf{P}^f_{i,t} (\mathbf{I} - \mathbf{K}_{i,t} \mathbf{V}_i)^T + \mathbf{K}_{i,t} \mathbf{R}_{i,t} \mathbf{K}_{i,t}^T$\label{step:7}
			
			\State \textbf{Output:} State estimate $\mathbf{s}^a_{i,t}$ and the error covariance $\mathbf{P}_{i,t}$
		\end{algorithmic}
	\end{algorithm}
	
	\subsubsection*{Complexity of Algorithm \ref{algo11}}
    For the length $m$ and $n$ of the respective vectors $\mathbf{s}_t$ and $\mathbf{z}_t$, a single Kalman filter update in \textbf{Algorithm}~\ref{algo11} has complexity $\mathcal{O}(n^3 + n^2 m + n m^2)$, dominated by covariance and gain computations, and reduces to $\mathcal{O}(n^3)$ when $m \le n$.

	\begin{algorithm}[!t]
		\caption{Improved active sensing using KF.}\label{algo2}
		\begin{algorithmic}[1]
			\State \textbf{System parameters:} $\{\mathbf{T}_i\}$, $\{\mathbf{V}_i\}$ for $i \in \{1, \cdots, P\}$, and others independent parameters
			\State \textbf{Input:} $\{\mathbf{s}_{i,0}^a\}$, $\{\mathbf{P}_{i,0}\}$
			\For{each $t \in \{1,\cdots, L\}$}
			\For{each $i \in \{1,\cdots, P\}$}
			\For{each $j \in \{0,\cdots, Z\}$}
			\State Determine $(\widehat{\tau}_{i,j,t},\widehat{\nu}_{i,j,t})$ and $\widehat{h}_{i,j,t}$ from \eqref{eq:tau_n_nu} and \eqref{eq:gain_est} using the known data $\mathbf{d}$\label{step:62}
			\State Determine the range $\widehat{\rho}_{i,j,t}$ and the radial velocity $\widehat{v}_{i,j,t}$ from \eqref{eq:est_tar_para}\label{step:72}
			\EndFor
			\State Generate $\mathbf{Q}_{i,t-1}$ and determine the selection based estimated observation vector $\mathbf{z}_{i,t} = [\widehat{\alpha}_{i,t},\widehat{\beta}_{i,t},\widehat{v}_{i,x,t},\widehat{v}_{i,y,t}]^T$ using \eqref{eq:nearest_sel} and calculate the noise covariance $\mathbf{R}_{i,t}$\label{step:92}
			\State Determine $\mathbf{s}_{i,t}^a$ and $\mathbf{P}_{i,t}$ using the function $f_{KF}(\mathbf{s}_{i,t-1}^a, \mathbf{P}_{i,t-1}$, $\mathbf{Q}(i,t-1)$, $\mathbf{R}_{i,t})$ in \textbf{Algorithm}~\ref{algo11}\label{step:102}
			\EndFor
			\EndFor
		\end{algorithmic}
	\end{algorithm}
	In \textbf{Algorithm}~\ref{algo2}, the active sensing as described in Section~\ref{sec:sys_mod} is further improved using the KF. Here, for each target $i$ and receiver $j$ at a given time $t$, the delay $\widehat{\tau}_{i,j}$, Doppler $\widehat{\nu}_{i,j}$, and the channel gain $\widehat{h}_{i,j}$ are estimated, followed by the range $\widehat{\rho}_{i,j}$ and the radial velocity $\widehat{v}_{i,j}$. Then, for each target $i$, the measurement vector $\mathbf{z}_{i,t}$ is determined; thereafter, using $f_{KF}$ in \textbf{Algorithm}~\ref{algo11}, the sensing variables are estimated.
	
	\subsubsection*{Complexity of Algorithm \ref{algo2}} 
	\textbf{Algorithm}~\ref{algo2} has complexity $\mathcal{O}\!\big(T P ((Z+1) G_\tau G_\nu MN + n^3)\big)$, 
dominated by 2D delay-Doppler correlation and Kalman updates, scaling linearly with $T$, $P$, and $Z$, quadratically with the delay--Doppler grid, and cubically with the state dimension $n$.

    \subsubsection*{Multistatic Integrated Sensing and Communication}
    The multistatic ISAC operates through an iterative refinement between coarse estimation, data detection, and channel update. Initially, coarse estimates of the delays, Doppler shifts, and channel gains are obtained from the pilot-based minimization in \eqref{eq:est_pilot_min} and transformed into the range and velocity domains using \eqref{eq:est_tar_para}. These parameters are then used to determine the target location and velocity through the selection approach in \eqref{eq:nearest_sel}, providing the initial conditions for iterative refinement. In each iteration, the received data are refined using the gradient-based update defined by \eqref{eq:gradient}, followed by symbol-wise demodulation as in \eqref{eq:symb_demod} to recover the transmitted information symbols. The reconstructed data are then used to re-estimate the channel parameters, and the process continues until the variations in channel gains, delays, and Doppler shifts fall below the predefined convergence thresholds\footnote{A detailed performance analysis of the KF-assisted multistatic ISAC within the considered geometric framework is presented in the extended Journal version of this work in \cite{rani20}.}.

	
	\section{Numerical Results and Conclusion}\label{sec:num_rel}
	This section evaluates the proposed analysis via numerical results. As per \textbf{Lemma}~\ref{lemma1}, a three-node topology with multiple antennas forming the largest triangular coverage outperforms randomly placed single-antenna nodes and is adopted for ISAC evaluation unless stated otherwise. System parameters are: $M=256$ subcarriers, $N=16$ symbols (4-QAM), $\Delta f=240\,\text{kHz}$, cyclic prefix 64, $f_c=30\,\text{GHz}$, SNR $0$\,dB; mobility: $\delta_i=1.5$, $\psi_i=0.5$; with $N_B=3$ nodes, $N_T=4$ targets, and trajectories updated every $\Delta t=0.5\,\text{s}$ in a $400\times400\,\text{m}^2$ area.

	\begin{figure}[!t]
		\centering\vspace{-3mm}\!\!\!\!
		\subfigure[]{\includegraphics[width=0.25\textwidth]{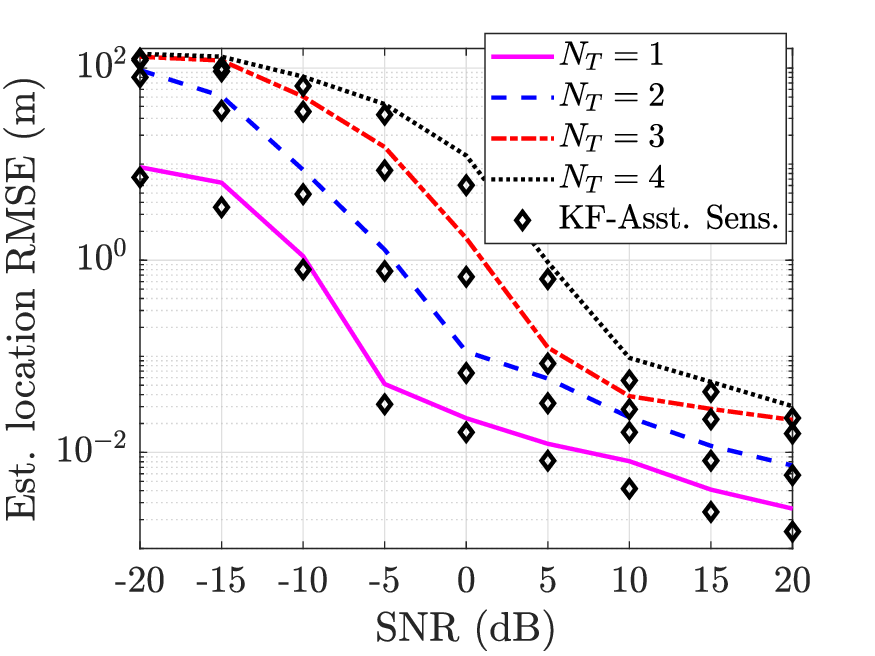}}\hspace{-0mm} \!\!\!\!
		\subfigure[]{\includegraphics[width=0.25\textwidth]{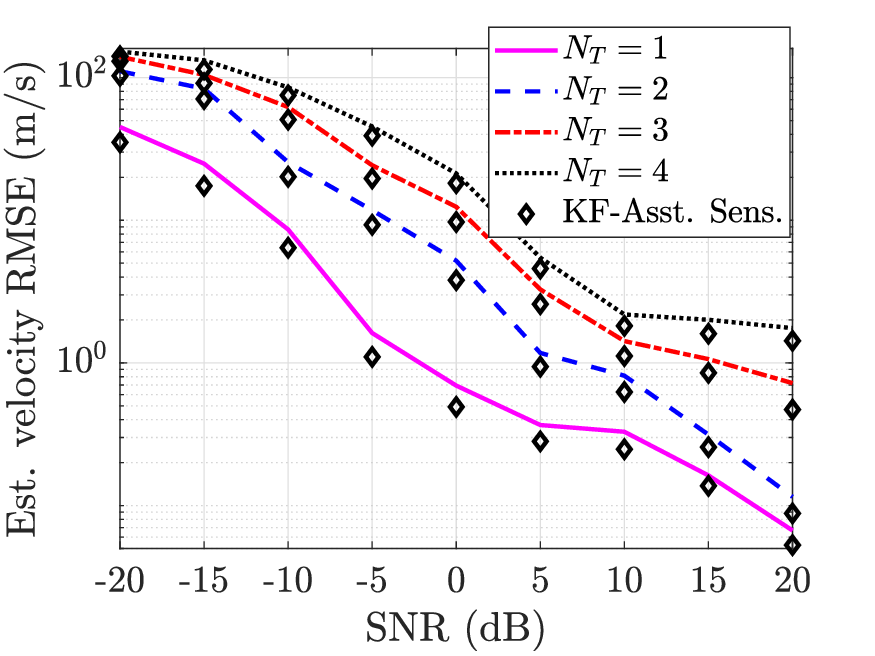}}\!\!\!\!
		\caption{RMSE improvement with SNR for different numbers of targets.}
		\label{fig:fig1}\vspace{-0mm}
	\end{figure} 

    \begin{figure}[!t]
		\centering\vspace{-3mm}\!\!\!\!
		\subfigure[]{\includegraphics[width=0.25\textwidth]{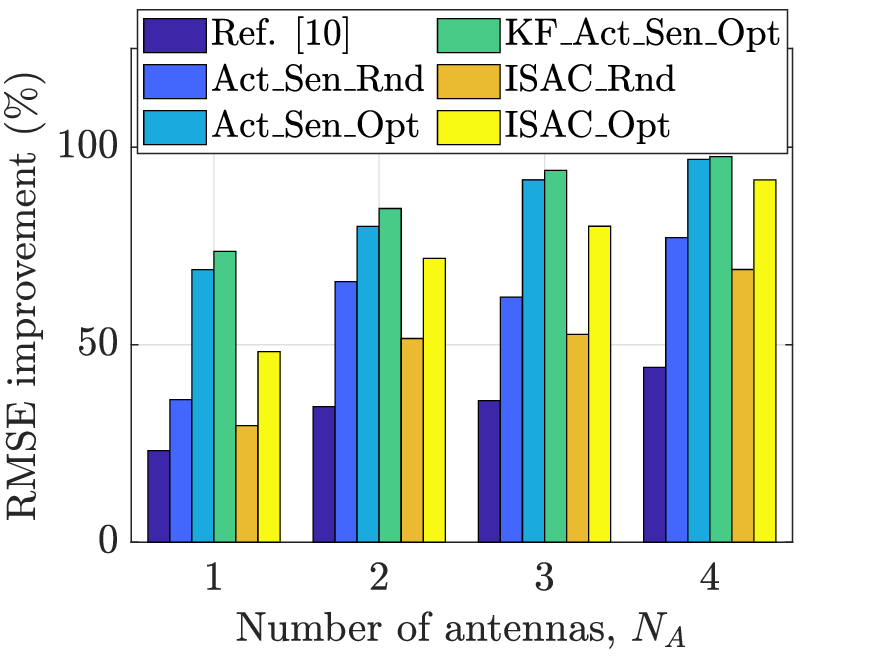}}\hspace{-0mm} \!\!\!\!
		\subfigure[]{\includegraphics[width=0.25\textwidth]{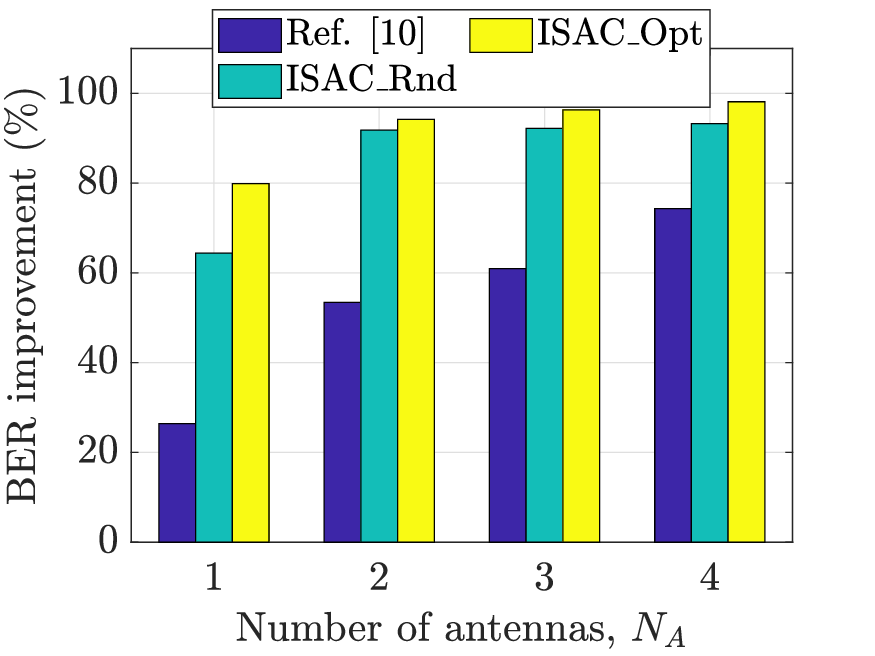}}\!\!\!\!
		\caption{Performance comparison of different schemes.}
		\label{fig:Scheme_Comp}\vspace{-5mm}
	\end{figure}

Using Fig.~\ref{fig:fig1}, the RMSE performance of active sensing and KF-assisted active sensing is compared for target localization and velocity estimation. The KF-assisted method consistently outperforms standard active sensing for all target counts by effectively balancing prediction and measurement noise, thus preventing overfitting. In Fig.~\ref{fig:fig1}(a), at $-20$~dB SNR, the RMSE rises sharply for multiple targets, increasing by approximately $85$~m for $N_T > 1$ due to the difficulty of distinguishing multiple peaks in noisy conditions. On average, RMSE increases by $15.53$~m, $18.00$~m, and $10.03$~m when $N_T$ changes from $1$ to $2$, $2$ to $3$, and $3$ to $4$, respectively, while KF-assisted sensing reduces RMSE by an average of $4.96$~m. Similarly, in Fig.~\ref{fig:fig1}(b), the velocity RMSE grows by $17.50$~m/s, $12.29$~m/s, and $10.87$~m/s for the same increments in $N_T$, whereas the KF-assisted method achieves an average improvement of $4.02$~m/s.

Fig.~\ref{fig:Scheme_Comp} compares multistatic OTFS--ISAC schemes, where sensing/communication gains are measured against their monostatic baselines under identical parameters. Act\_Sen\_Rnd and Act\_Sen\_Opt perform active sensing, while ISAC\_Rnd and ISAC\_Opt enable ISAC with random and optimized receiver deployment. KF\_Act\_Sen\_Opt incorporates KF-assisted tracking with optimized geometry. In Fig.~\ref{fig:Scheme_Comp}(a), optimized schemes markedly lower RMSE, with KF\_Act\_Sen\_Opt achieving the best accuracy via spatial-temporal information. In Fig.~\ref{fig:Scheme_Comp}(b), ISAC\_Rnd and ISAC\_Opt achieve $85.4\%$ and $92.1\%$ BER gains, respectively, outperforming the $53.8\%$ improvement in \cite{sruti18}, confirming that OTFS and geometry-aware deployment significantly enhance sensing and communication.

\section*{Acknowledgment}
In this paper, the work of Z. Liu was supported in part by the UK Engineering and Physical Sciences Research Council under Grants EP/X035352/1 (\textit{DRIVE}), EP/Y000986/1 (\textit{SORT}), and EP/Y037243/1 (\textit{REVOL6G}), and by the Royal Society under Grants IEC\textbackslash NSFC\textbackslash 233292 and IES\textbackslash R1\textbackslash 241212. Additionally, this work was also supported by DST-funded INAE Indo-Taiwan project 2023/INTW/10.

\vspace{-1mm}

	\vspace{-1mm}
	
	\makeatletter
	\renewenvironment{thebibliography}[1]{%
		\@xp\section\@xp*\@xp{\refname}%
		\normalfont\footnotesize\labelsep .5em\relax
		\renewcommand\theenumiv{\arabic{enumiv}}\let\p@enumiv\@empty
		\vspace*{-1pt}
		\list{\@biblabel{\theenumiv}}{\settowidth\labelwidth{\@biblabel{#1}}%
			\leftmargin\labelwidth \advance\leftmargin\labelsep
			\usecounter{enumiv}}%
		\sloppy \clubpenalty\@M \widowpenalty\clubpenalty
		\sfcode`\.=\@m
	}{%
		\def\@noitemerr{\@latex@warning{Empty `thebibliography' environment}}%
		\endlist
	}
	\makeatother
	\bibliographystyle{IEEEtran}
	\bibliography{references_COMML}
\end{document}